\documentclass[11pt]{amsart}

\usepackage{latexsym}
\usepackage{amsmath}
\usepackage{amssymb}
\usepackage{xcolor}
\usepackage{enumerate}
\usepackage{graphicx}

\newtheorem{theorem}{Theorem}[section]
\newtheorem{lemma}[theorem]{Lemma}
\newtheorem{corollary}[theorem]{Corollary}

\newcommand{\cT}{{\mathcal T}}

\begin{document}

\title[Phylogenetic Diversity and Diversity Indices]{Quantifying the Difference Between Phylogenetic Diversity and Diversity Indices}

\thanks{The second author was supported by the New Zealand Marsden Fund}

\author{Magnus Bordewich}
\address{Department of Computer science, Durham University, United Kingdom}
\email{m.j.r.bordewich@durham.ac.uk}

\author{Charles Semple}
\address{School of Mathematics and Statistics, University of Canterbury, Christchurch, New Zealand}
\email{charles.semple@canterbury.ac.nz}



\date{\today}

\maketitle

\begin{abstract}
Phylogenetic diversity is a popular measure for quantifying the biodiversity of a collection $Y$ of species, while phylogenetic diversity indices provide a way to apportion phylogenetic diversity to individual species. Typically, for some specific diversity index, the phylogenetic diversity of $Y$ is not equal to the sum of the diversity indices of the species in $Y.$ In this paper, we investigate the extent of this difference for two commonly-used indices: Fair Proportion and Equal Splits. In particular, we determine the maximum value of this difference under various instances including when the associated rooted phylogenetic tree is allowed to vary across all root phylogenetic trees with the same leaf set and whose edge lengths are constrained by either their total sum or their maximum value.
\end{abstract}

\section{Introduction}

Phylogenetic diversity (PD) is prominent measure in evolutionary biology to quantify the biological diversity of a collection of species. Intuitively, the phylogenetic diversity of such a collection quantifies how much of the `Tree of Life' is spanned by the species in the collection. Introduced by Faith in 1992~\cite{fai92}, PD has been analysed and applied in a variety of contexts for various taxa including plants, bacteria, and mammals~\cite{cad08, loz07, saf11}.

On the other hand, (phylogenetic) diversity indices, also called evolutionary distinctiveness measures, quantify an individual species' contribution to overall phylogenetic diversity, thus providing a convenient way to rank species as, for example, in conservation planning~\cite{red08, red14}. However, although the sum of the diversity indices across all species equates to the total phylogenetic diversity, in general, the sum of the diversity indices of a subset of the species differs from the subset's phylogenetic diversity. In this paper, we investigate the extent of this difference for two natural and well-known diversity indices, Fair Proportion and Equal Splits. Both of these indices are used, for example, in the conservation initiative `EDGE of Existence Programme' established by the Zoological Society of London~\cite{isa07}. A version of this investigation was {posed} in earlier work~\cite{haa08} {but}, apart from a related simulation study~\cite{red08}, it has not been explored further. Furthermore, the work in this paper is related to a recent study in~\cite{wic20} in which the authors consider the extent of the difference between the Fair Proportion and Equal Splits indices of a species. We next formalise the investigation in this paper, ending the introduction with a high-level overview of the main results.

Throughout the paper, $X$ is a non-empty finite set. A {\em rooted phylogenetic $X$-tree} $\cT$ is a rooted tree with leaf set $X$ whose non-leaf vertices have out-degree at least two. For technical reasons, if $|X|=1$, we additionally allow a rooted phylogenetic tree to consist of the single vertex in $X$. If all non-leaf vertices have out-degree exactly two, then $\cT$ is {\em binary}. For ease of reading, as all phylogenetic trees in this paper are rooted and binary, we will refer to a ``rooted binary phylogenetic tree'' as simply a ``phylogenetic tree''.

Let $\cT$ be a phylogenetic $X$-tree with root $\rho$, and consider a map $\ell: E\rightarrow \mathbb R^{\ge 0}$ from the set $E$ of edges of $\cT$ to the non-negative reals. Collectively, we denote $\cT$ and $\ell$ by the ordered pair $(\cT, \ell)$ and refer to $(\cT, \ell)$ as an {\em edge-weighted phylogenetic $X$-tree}. Furthermore, we use $L(\cT, \ell)$ to denote the total sum of the edge lengths of $(\cT, \ell)$. The {\em phylogenetic diversity} of a subset $Y$ of $X$ on $(\cT, \ell)$, denoted $PD_{(\cT, \ell)}(Y)$, is the sum of the weights of the edges in the (unique) minimal subtree of $\cT$ connecting the vertices in $Y\cup \{\rho\}$. Observe that $PD_{(\cT, \ell)}(X)=L(\cT, \ell)$. To illustrate, consider the edge-weighted phylogenetic $X$-tree shown in Fig.~\ref{tree}, where $X=\{x_1, x_2, \ldots, x_7\}$. The minimal subtree of $\cT$ connecting the vertices in $\{x_2, x_4, x_5\}\cup \{\rho\}$ is indicated by dashed edges, and so PD of $\{x_2, x_4, x_5\}$ is $18$. The PD of $\{x_5, x_7\}$ is $10$.

\begin{figure}
\center
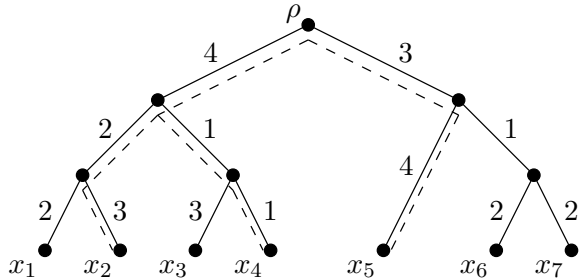
\caption{An edge-weighted phylogenetic $X$-tree $(\cT, \ell)$. The minimal subtree of $\cT$ connecting the vertices in $\{x_2, x_4, x_5\}\cup \{\rho\}$ is indicated with dashed edges.}
\label{tree}
\end{figure}

A {\em phylogenetic diversity index} for an edge-weighted phylogenetic $X$-tree $(\cT, \ell)$ is a function $\varphi_{(\cT, \ell)}: X\rightarrow \mathbb R^{\ge 0}$ that assigns a score to each leaf of $\cT$ such that
$$\sum_{x\in X} \varphi_{(\cT, \ell)}(x)=PD_{(\cT, \ell)}(X)=L(\cT, \ell).$$
Furthermore, if, for each $x\in X$, we can write $\varphi_{(\cT, \ell)}(x)$ as a linear function of the edge lengths of $\cT$, that is,
\begin{align}
\varphi_{(\cT, \ell)}(x)=\sum_{e\in E} \gamma_{(\cT, \ell)}(x, e)\cdot \ell(e)
\label{linear1}
\end{align}
for some constants $\gamma_{(\cT, \ell)}(x, e)$ that are independent of $\ell(e)$, we say $\varphi_{(\cT, \ell)}$ is a {\em linear diversity index}. It is easily checked that an arbitrary function $\varphi_{(\cT, \ell)}$ of the form shown in~(\ref{linear1}) is a phylogenetic diversity index if and only if, for each edge $e$ of $\cT$, we have
\begin{align}
\sum_{x\in X} \gamma_{(\cT, \ell)}(x, e)=1.
\label{sum1}
\end{align}

Two well-studied linear diversity indices underlie the results in this paper. Let $(\cT, \ell)$ be an edge-weighted phylogenetic $X$-tree. The {\em Fair Proportion index} (FP) for a leaf $x\in X$, denoted $FP_{(\cT, \ell)}(x)$, is the value
$$FP_{(\cT, \ell)}(x)=\sum_{e\in P(\cT; \rho, x)} \frac{1}{n(e)}\ell(e),$$
where $P(\cT; \rho, x)$ denotes the (unique) path in $\cT$ from the root $\rho$ to $x$, and $n(e)$ denotes the number of leaves that are at the end of a directed path starting at the root and traversing $e$. Intuitively, the FP index distributes the length of $e$ evenly amongst its descendant leaves. The Fair Proportion index is also called `evolutionary distinctiveness'~\cite{isa07}. The second index is the {\em Equal Splits index} (ES) which for a leaf $x\in X$ is the value
$$ES_{(\cT, \ell)}(x)=\sum_{e\in P(\cT; \rho, x)} \frac{1}{\Pi(e, x)}\ell(e),$$
where $\Pi(e, x)=1$ if $e$ is the pendant edge of $\cT$ incident with $x$, and if $e=(u, v)$ is a non-pendant edge of $\cT$, then $\Pi(e, x)$ is the product of the out-degrees of the non-leaf vertices (including $v$) on the directed path from $v$ to $x$. Since all phylogenetic trees in this paper are binary, $\Pi(e, x)$ is always a power of $2$. In particular, if $e=(u, v)$ is a non-pendant edge of $\cT$, then $\Pi(e, x)$ is $2^m$, where $m$ is the number of edges from $v$ to $x$. The Fair Proportion and Equal Splits indices were introduced in~\cite{red03} and~\cite{red06}, respectively, and a direct comparison of these two indices was investigated recently in~\cite{wic20}.

As examples of the Fair Proportion and Equal Splits indices, consider the edge-weighted phylogenetic $X$-tree in Fig.~\ref{tree}. The FP indices for $x_5$ and $x_7$ are
$$FP_{(\cT, \ell)}(x_5) = {\textstyle \frac{1}{3}\cdot 3 + 1\cdot 4 = 5}$$
and
$$FP_{(\cT, \ell)}(x_7) = {\textstyle \frac{1}{3}\cdot 3 + \frac{1}{2}\cdot 1 + 1\cdot 2 = 3\frac{1}{2}}.$$
The ES indices for $x_5$ and $x_7$ are
$$ES_{(\cT, \ell)}(x_5) = {\textstyle \frac{1}{2}\cdot 3 + 1\cdot 4 = 5\frac{1}{2}}$$
and
$$ES_{(\cT, \ell)}(x_7) = {\textstyle \frac{1}{4}\cdot 3 + \frac{1}{2}\cdot 1 + 1\cdot 2 = 3\frac{1}{4}}.$$

For an edge-weighted phylogenetic $X$-tree $(\cT, \ell)$ and diversity index $\varphi$, although the sum of the diversity indices across all taxa in $X$ equates to $PD_{(T, \ell)}(X)$, the sum of the diversity indices across all taxa in a proper subset $Y$ of $X$ will typically not equal $PD_{(\cT, \ell)}(Y)$. For example, in Fig.~\ref{tree}, $PD_{(\cT, \ell)}(\{x_5, x_7\})=10$, but
$$\sum_{x\in \{x_5, x_7\}} FP_{(\cT, \ell)}(x) = {\textstyle 5 + 3\frac{1}{2} = 8\frac{1}{2}}.$$
To quantify this, set
\begin{align}
\Delta_{(\varphi, \cT, \ell)}(Y)=PD_{(\cT, \ell)}(Y)-\sum_{x\in Y} \varphi_{(\cT, \ell)}(x).
\label{delta}
\end{align} 
We refer to $\Delta_{(\varphi, \cT, \ell)}(Y)$ as the {\em diversity difference of $Y$} (relative to $(\cT, \ell)$ and $\varphi$). In this paper, we are interested in the extent of this difference for the Fair Proportion and Equal Splits indices.

The following lemma shows that the diversity difference is always non-negative provided $\varphi$ satisfies a natural distribution property. This property says that, for an edge-weighted phylogenetic $X$-tree $(\cT, \ell)$ and all $x\in X$, the value $\varphi_{(\cT, \ell)}(x)$ can be written as a non-negative linear function of the lengths of edges in the path from the root $\rho$ of $\cT$ to $x$.

\begin{lemma}
Let $(\cT, \ell)$ be an edge-weighted phylogenetic $X$-tree with root $\rho$, and let $\varphi_{(\cT, \ell)}$ be a linear diversity index such that, for each $x\in X$,
\begin{align}
\varphi_{(\cT, \ell)}(x)=\sum_{e\in P(\cT; \rho, x)} \gamma_{(\cT, \ell)}(x, e)\cdot \ell(e)
\label{linear2}
\end{align}
for some non-negative constants $\gamma_{(\cT, \ell)}(x, e)$ that are independent of $\ell(e)$. If $Y\subseteq X$, then
$$\Delta_{(\varphi, \cT, \ell)}(Y)=PD_{(\cT, \ell)}(Y)-\sum_{x\in Y} \varphi_{(\cT, \ell)}(x)\ge 0.$$
\label{non-negative}
\end{lemma}

\begin{proof}
Let $Y\subseteq X$. Then
$$PD_{(\cT, \ell)}(Y)=\sum_{e\in \cT(Y\cup \{\rho\})} \ell(e)$$
and
$$\sum_{x\in Y} \varphi_{(\cT, \ell)}(x)=\sum_{e\in \cT(Y\cup \{\rho\})} \left(\sum_{x\in Y} \gamma_{(\cT, \ell)}(x, e)\right)\cdot \ell(e),$$
where $\cT(Y\cup \{\rho\})$ denotes the minimal subtree of $\cT$ connecting the leaves in $Y$ and $\rho$. Since $\gamma_{(\cT, \ell)}(x, e)$ is non-negative for all $x\in X$ and edges $e$ of $\cT$, it follows by (\ref{sum1}) that $\sum_{x\in Y} \gamma_{(\cT, \ell)}(x, e)\le 1$. Hence $\Delta_{(\varphi, \cT, \ell)}(Y)\ge 0$ as required.
\end{proof}

We call a phylogenetic diversity index that satisfies (\ref{linear2}) a {\em descendant diversity index}.  It is easily checked that the Fair Proportion and Equal Splits indices are examples of descendant diversity indices.

In this paper, for each of the Fair Proportion and Equal Splits indices, we will determine, for all positive integers $k$, the maximum value of the diversity difference over all subsets of $X$ of size $k$ under the following instances:
\begin{enumerate}[(i)]
\item a fixed phylogenetic $X$-tree whose edge lengths are fixed;

\item a fixed phylogenetic $X$-tree whose edge lengths are constrained by either their (total) sum or their maximum value; and

\item across all phylogenetic $X$-trees whose edge lengths are constrained by either their (total) sum or their maximum value.
\end{enumerate}
In particular, for (i) and (ii), we give polynomial-time algorithms for finding these maximum values for an arbitrary descendant diversity index while, for (iii), we characterise the edge-weight phylogenetic $X$-trees and subsets of size $k$ that realise this maximum value for FP and ES. For (iii), it turns out that, in the case that the edge lengths are constrained by their sum, the class of phylogenetic trees that maximise the diversity difference under ES is always a subclass of the phylogenetic trees that maximise the diversity difference under FP. The corresponding characterisations are stated as Theorems~\ref{across-sum-es} and~\ref{across-sum-fp}, respectively. However, in the case that the edge lengths are constrained by their maximum value, the class of phylogenetic trees that maximise the diversity difference under FP and ES coincide (Theorem~\ref{across-max-fp}).

The paper is organised as follows. The next section contains some preliminaries that are used throughout the paper. Section~\ref{fixed-tree} considers maximising the diversity difference on a fixed phylogenetic tree, while Section~\ref{across-trees} considers maximising the diversity difference across all phylogenetic trees. The last section, Section~\ref{discussion}, consists of a brief discussion.

\section{Preliminaries}

Let $\cT$ be a phylogenetic $X$-tree, and let $t$ and $w$ be vertices of $\cT$. If $P$ is a (directed) path in $\cT$ from $t$ to $w$, the length of $P$, denoted $|P|$, is the number of edges in $P$. We denote this path by $P(\cT; t, w)$. Furthermore, we sometimes refer to a path in $\cT$ from an edge $e=(u, v)$ to $w$, in which case, we mean from $v$ to $w$, and denote this path by $P(\cT; e, w)$.

\noindent {\bf Subtrees.} Let $\cT$ be a phylogenetic $X$-tree with root $\rho$, and let $X'$ be a subset of $X\cup \{\rho\}$. The minimal subtree of $\cT$ connecting the vertices in $X'$ is denoted by $\cT(X')$. Furthermore, if $\rho\not\in X'$, the {\em restriction} of $\cT$ to $X'$, denoted by $\cT|X'$, is the phylogenetic $X'$-tree obtained from $\cT(X')$ by suppressing all non-root vertices of degree two.

Let $\cT$ be a phylogenetic $X$-tree and let $e=(u, v)$ be an edge of $\cT$. The subset of leaves that are descendants of $v$ is called a {\em cluster} of $\cT$ and is denoted by $C(v)$. The set $X$ as well as each of the singleton subsets of $X$ are clusters of every phylogenetic $X$-tree, and so such clusters are called the {\em trivial} clusters of $\cT$. Hence a cluster $X'$ of $\cT$ is {\em non-trivial} if $2\le |X'|< |X|$. Thus if $\cT$ has a non-trivial cluster, then $|X|\ge 3$. Furthermore, the rooted subtree obtained from $\cT$ by deleting $e$ and whose root is $v$ is a {\em pendant subtree} of $\cT$.

\noindent {\bf Cherries and chains.} A $2$-element subset of $X$, say $\{a, b\}$, is a {\em cherry} of a phylogenetic $X$-tree $\cT$ if $a$ and $b$ have the same parent. Now let $X'=\{a_1, a_2, \ldots, a_m\}$ be a subset of $X$ such that $m\ge 2$ and, for all $i\in \{1, 2, \ldots, m\}$, let $p_i$ denote the parent of $a_i$. We call $X'$ a {\em chain} of $\cT$ if there is an ordering of the elements of $X'$, say $(a_1, a_2, \ldots, a_m)$, such that
$$p_m,\, p_{m-1},\, \ldots,\, p_1$$
is a (directed) path in $\cT$, in which case, $p_m$ is the {\em first parent} of the chain and $p_1$ is the {\em last parent} of the chain. Note that we still refer to $X'$ as a chain of $\cT$ if $\{p_1, p_2\}$ is a cherry and $p_m, p_{m+1}, \ldots, p_2=p_1$ is a path in $\cT$. The {\em edge set of $X'$} consists of the pendant edges incident with a leaf in $X'$ as well as the edges in the path $p_m,\, p_{m-1},\, \ldots,\, p_1$.

\noindent {\bf Subtree prune and regraft.} Let $\cT$ be a phylogenetic $X$-tree. {For the purposes of the upcoming operation,} view the root $\rho$ of $\cT$ adjoined to the original root via a pendant edge. Let $(u, v)$ be an edge of $\cT$ such that $u\neq \rho$. Let $\cT'$ be the phylogenetic $X$-tree obtained from $\cT$ by deleting $(u, v)$ and suppressing $u$, and then reattaching the pendant subtree $\cT(C(v))$ by subdividing an edge, $f$ say, in the component of $\cT\backslash (u, v)$ containing $\rho$ with a new vertex $u'$ and adjoining $\cT(C(v))$ with a new edge $(u', v)$. We say that $\cT'$ has been obtained from $\cT$ by a {\em rooted subtree prune and regraft} operation. More specifically, in this operation, we have {\em pruned} $C(v)$ and {\em regrafted} it to $f$. In the special case $f$ is the pendant edge incident with the root of $\cT$, we say that $C(v)$ is pruned and {\em regrafted to $\rho$}. At the completion of this operation, we no longer view the root $\rho$ as being joined via a pendant edge, and so the root of $\cT'$ (labelled $\rho$) is the unique vertex of $\cT'$ of in-degree zero after deleting the temporary root and its incident edge.

\noindent {\bf Diversity indices.} Let $(\cT, \ell)$ be an edge-weighted phylogenetic $X$-tree with edge set $E$, and let $\varphi$ be a descendant diversity index. Then, for each $x\in X$,
$$\varphi_{(\cT, \ell)}(x)=\sum_{e\in P(\cT; \rho, x)} \gamma_{(\cT, \ell)}(x, e)\cdot \ell(e)$$
for some non-negative constants $\gamma_{(\cT, \ell)}(x, e)$ that are independent of $\ell(e)$. For a subset $Y$ of $X$, we denote the contribution of an edge $e=(u, v)$ of $\cT$ to $\Delta_{(\varphi, \cT, \ell)}(Y)$ by $\lambda_{(\cT, \ell)}(e)$, that is, 
$$\lambda_{(\cT, \ell)}(e)=\ell(e)-\sum_{y\in Y\cap C(v)} \gamma_{(\cT, \ell)}(y, e)\cdot \ell(e)= l(e)\left( 1- \sum_{y\in Y\cap C(v)} \gamma_{(\cT, \ell)}(y, e)\right)$$
if $Y\cap C(v)$ is non-empty, and $\lambda_{(\cT, \ell)}(e)=0$ otherwise. Thus
$$\sum_{e\in E} \lambda_{(\cT, \ell)}(e)=\Delta_{(\varphi, \cT, \ell)}(Y).$$
It will be clear in the context which subset of $X$ and which descendant diversity index $\lambda_{(\cT, \ell)}(e)$ is referring to. Observe that $\lambda_{(\cT, \ell)}(e)\ge 0$. Moreover, if either $C(v)\subseteq Y$ or $C(v)\subseteq X-Y$, then $\lambda_{(\cT, \ell)}(e)=0$.
$$\lambda_{(\cT, \ell')}(e)=\ell(e)+t-\sum_{y\in Y\cap C(v)} \gamma_{(\cT, \ell)}(y, e)\cdot (\ell(e)+t).$$

\section{Maximising the Diversity Difference on a Fixed Tree}
\label{fixed-tree}

In this section, we consider the problem of determining the maximum value of the diversity difference for a fixed phylogenetic tree for when $\varphi$ is an arbitrary descendant diversity index.

\subsection{Fixed edge lengths}

We first consider the instance of a fixed phylogenetic tree with fixed edge lengths. Let $(\cT, \ell)$ be an edge-weighted phylogenetic $X$-tree, and let $\varphi$ be an arbitrary linear diversity index. Now let $(\cT, \ell')$ be the edge-weighted phylogenetic $X$-tree obtained from $(\cT, \ell)$ by adding, for each $x\in X$, the value
$$L(\cT, \ell)-\varphi_{(\cT, \ell)}(x)$$
to the length of the pendant edge incident with $x$. For all non-pendant edges $e$, we have $\ell'(e)=\ell(e)$. Note that the choice of $L(\cT, \ell)$ is not unique, we simply require a constant value that is large enough so that all pendant edges have non-negative lengths.

\begin{lemma}
Let $(\cT, \ell)$ be an edge-weighted phylogenetic $X$-tree, let $\varphi$ be a linear diversity index, let $k$ be a non-negative integer, and let $Y$ be a subset of $X$ of size $k$. Then
$$\Delta_{(\varphi, \cT, \ell)}(Y)=PD_{(\cT, \ell')}(Y)-k\cdot L(\cT, \ell).$$
\label{modify}
\end{lemma}

\begin{proof}
Now,
\begin{align*}
\Delta_{(\varphi, \cT, \ell)}(Y) & = PD_{(\cT, \ell)} (Y) - \sum_{x\in Y} \varphi_{(\cT, \ell)}(x) \\
& = \left(PD_{(\cT, \ell')}(Y) - \sum_{x\in Y} \left(L(\cT, \ell) - \varphi_{(\cT, \ell)}(x)\right)\right) - \sum_{x\in Y} \varphi_{(\cT, \ell)}(x) \\
& = PD_{(\cT, \ell')}(Y) -k\cdot L(\cT, \ell) + \sum_{x\in Y} \varphi_{(\cT, \ell)}(x) - \sum_{x\in Y} \varphi_{(\cT, \ell)}(x) \\
& = PD_{(\cT, \ell')}(Y) - k\cdot L(\cT, \ell).
\end{align*}
\end{proof}

Lemma~\ref{modify} is the basis of the following algorithm which finds the maximum value of the diversity difference under a linear diversity index for a given edge-weighted phylogenetic tree. It is well known that Step~2 of the algorithm, that is, finding a subset $Y$ of $X$ of size $k$ such that $PD_{(\cT, \ell')}(Y)\ge PD_{(\cT, \ell')}(Y')$ for all subsets $Y'$ of $X$ of size $k$, takes quadratic time in the size of $X$~\cite{par05, ste05}.

\noindent {\sc MaxDiversityDiff} \\
\noindent {\bf Input:} An edge-weighted phylogenetic $X$-tree $(\cT, \ell)$, a linear diversity index $\varphi$, and a non-negative integer $k$. \\
\noindent {\bf Output:} A subset $Y$ of $X$ of size $k$ that maximises $\Delta_{(\varphi, \cT, \ell)}(Y)$ across all subsets of $X$ of size $k$.
\begin{enumerate}[1.]
\item Let $\ell': E(\cT)\rightarrow \mathbb R^{\ge 0}$ be the map defined by setting $\ell'(e)=L(\cT, \ell)-\varphi_{(\cT, \ell)}(x)$ for all pendant edges $e$, where $x$ is the leaf incident with $e$, and $\ell'(e)=\ell(e)$ for all non-pendant edges. Construct $(\cT, \ell')$.

\item Find a subset $Y$ of $X$ of size $k$ such that $PD_{(\cT, \ell')}(Y)\ge PD_{(\cT, \ell')}(Y')$ for all subsets $Y'$ of $X$ of size $k$.

\item Return $Y$ and $\Delta_{(\varphi, T, \ell)}(Y)=PD_{(\cT, \ell')}(Y)-k\cdot L(\cT, \ell)$.
\end{enumerate}

The next theorem is an immediate consequence of the discussion prior to the description of {\sc MaxDiversityDiff} and the fact that the construction of $(\cT, \ell')$ from $(\cT, \ell)$ takes linear time (in the size of $X$).

\begin{theorem}
Let $(\cT, \ell)$ be an edge-weighted phylogenetic $X$-tree, let $k$ be a positive integer, and let $\varphi$ be a linear diversity index. Then applying {\sc MaxDiversityDiff} to $(\cT, \ell)$, $\varphi$, and $k$ correctly returns a subset $Y$ of $X$ of size $k$ that maximises
$$\Delta_{(\varphi, \cT, \ell)}(Y)=PD_{(\cT, \ell)}(Y)-\sum_{x\in Y} \varphi_{(\cT, \ell)}(x)$$
across all subsets of $X$ of size $k$ in time $O\left(|X|^2\right)$.
\label{poly-time}
\end{theorem}

\subsection{Maximum sum of edge lengths}

Now consider the problem of maximising the diversity difference on a fixed phylogenetic tree whose edge lengths are constrained by their (total) sum. We will make use of the next lemma in this subsection as well as in Section~\ref{across-trees}.

\begin{lemma}
Let $k$ be a positive integer, let $\cT$ be a phylogenetic $X$-tree, let $m$ be a positive real, and let $\varphi$ be a descendant diversity index. Let $\ell: E\rightarrow \mathbb R^{\ge 0}$ be a map on the set $E$ of edges of $\cT$ with $L(\cT, \ell)=m$ and let $Y$ be a subset of $X$ of size $k$ that maximises
$$\Delta_{(\varphi, \cT, \ell)}(Y)=PD_{(\cT, \ell)}(Y)-\sum_{x\in Y} \varphi_{(\cT, \ell)}(x)$$
across all edge weightings $\ell'$ of $\cT$ with $L(\cT, \ell')=m$ and all subsets of $X$ of size $k$. Suppose that $e_1=(u_1, v_1)$ and $e_2=(u_2, v_2)$ are distinct edges of $\cT$ such that $\ell(e_1) > 0$ and $\ell(e_2) > 0$. Then
\begin{enumerate}[{\rm (i)}]
\item $\sum_{y\in Y\cap C(v_1)} \gamma_{(\cT, \ell)}(y, e_1)=\sum_{y\in Y\cap C(v_2)} \gamma_{(\cT, \ell)}(y, e_2).$

\item $\Delta_{(\varphi, \cT, \ell')}(Y)=\Delta_{(\varphi, \cT, \ell)}(Y)$, where $\ell'(e_1)=\ell(e_1)+\ell(e_2)$, $\ell'(e_2)=0$, and $\ell'(f)=\ell(f)$ for all $f\in E-\{e_1, e_2\}$.
\end{enumerate}
\label{new-optimal}
\end{lemma}

\begin{proof}
Let $\ell: E\rightarrow \mathbb R^{\ge 0}$ be such a map. For each $i\in \{1, 2\}$, let $e_i=(u_i, v_i)$ and, for each edge $e$ of $\cT$, define $\lambda_{(\cT, \ell)}(e)$ and $\lambda_{(\cT, \ell')}(e)$ to be the contribution of $e$ to $\Delta_{(\varphi, \cT, \ell)}(Y)$ and $\Delta_{(\varphi, \cT, \ell')}(Y)$, respectively. To prove the lemma, it suffices to show that
$$\lambda_{(\cT, \ell)}(e_1)+\lambda_{(\cT, \ell)}(e_2) = \lambda_{(\cT, \ell')}(e_1)+\lambda_{(\cT, \ell')}(e_2).$$
Since $\ell'(e_2)=0$, it follows that $\lambda_{(\cT, \ell')}(e_2)=0$. Now, for each $i\in \{1, 2\}$,
$$\lambda_{(\cT, \ell)}(e_i)=\ell(e_i)-\sum_{y\in Y\cap C(v_i)}\gamma_{(\cT, \ell)}(y, e_i)\cdot \ell(e_i).$$
If
$$\sum_{y\in Y\cap C(v_1)}\gamma_{(\cT, \ell)}(y, e_1) < \sum_{y\in Y\cap C(v_2)}\gamma_{(\cT, \ell)}(y, e_2),$$
then, by linearity,
\begin{align*}
\lambda_{(\cT, \ell)}(e_1)+\lambda_{(\cT, \ell)}(e_2) & < \ell(e_1)+\ell(e_2)-\sum_{y\in Y\cap C(v_1)}\gamma_{(\cT, \ell)}(y, e_1)\cdot (\ell(e_1)+\ell(e_2)) \\
& = \lambda_{(\cT, \ell')}(e_1)+\lambda_{(\cT, \ell')}(e_2),
\end{align*}
a contradiction to the maximality of $\ell$. Using a symmetric argument, it follows that
$$\sum_{y\in Y\cap C(v_1)}\gamma_{(\cT, \ell)}(y, e_1)=\sum_{y\in Y\cap C(v_2)}\gamma_{(\cT, \ell)}(y, e_2).$$
Hence
\begin{align*}
\lambda_{(\cT, \ell)}(e_1)+\lambda_{(\cT, \ell)}(e_2) & = \ell(e_1) + \ell(e_2) - \sum_{y\in Y\cap C(v_1)} \gamma_{(\cT, \ell)}(y, e_1)\cdot (\ell(e_1) + \ell(e_2)) \\
& = \lambda_{(\cT, \ell')}(e_1).
\end{align*}
This completes the proof of the lemma.
\end{proof}

By repeated applications of Lemma~\ref{new-optimal}(ii), we obtain the following corollary.

\begin{corollary}
Let $\cT$ be a phylogenetic $X$-tree, let $k$ be a positive integer, let $m$ be a positive real, and let $\varphi$ be a descendant diversity index. Then there exists a weighting $\ell: E\rightarrow \mathbb R^{\ge 0}$ of the edge set $E$ of $\cT$ with $\ell(e)=m$ for some edge $e\in E$ and a subset $Y$ of $X$ of size $k$ that maximises
$$\Delta_{(\varphi, \cT, \ell)}(Y)=PD_{(\cT, \ell)}(Y)-\sum_{x\in Y} \varphi_{(\cT, \ell)}(x)$$
across all edge weightings $\ell'$ of $\cT$ with $L(\cT, \ell')=m$ and all subsets of $X$ of size $k$.
\label{single-edge}
\end{corollary}

It follows by Corollary~\ref{single-edge} that, given $\cT$, $k$, $m$, and $\varphi$ as in its statement, we can find an edge wighting $\ell: E\rightarrow \mathbb R^{\ge 0}$ that maximises $\Delta_{(\varphi, \cT, \ell)}(Y)$ across all subsets of $X$ of size $k$ as follows. For all $e\in E$, let $(\cT, \ell_e)$ denote the edge-weighted phylogenetic $X$-tree, where $\ell_e$ is the edge weighting of $\cT$ in which $\ell(e)=m$ and $\ell(f)=0$ for all $f\in E-\{e\}$. Now, for each $e\in E$, apply Theorem~\ref{poly-time} and, more particularly, {\sc MaxDiversityDiff} to $(\cT, \ell_e)$, $k$, and $\varphi$. The maximum of the values returned by these applications gives the desired value (as well as a subset of $X$ of size $k$ realising this value). Since the total number of edges in $\cT$ is $2|X|-2$, we have the next theorem.

\begin{theorem}
Let $\cT$ be a phylogenetic $X$-tree, let $k$ be a positive integer, let $m$ be a positive real, and let $\varphi$ be a descendant diversity index. Then finding an edge weighting $\ell$ of $\cT$ with $L(\cT, \ell)=m$ and a subset $Y$ of $X$ of size $k$ that maximises
$$\Delta_{(\varphi, \cT, \ell)}(Y)=PD_{(\cT, \ell)}(Y)-\sum_{x\in Y} \varphi_{(\cT, \ell)}(x)$$
across all edge weightings $\ell'$ of $\cT$ with $L(\cT, \ell')=m$ and all subsets of $X$ of size $k$ takes time $O\left(|X|^3\right)$.
\end{theorem}

\subsection{Maximum edge length}

Next consider the problem of maximising the diversity difference on a fixed phylogenetic $X$-tree $\cT$ whose edge lengths are constrained by some maximum value, say $\ell_{\max}$. We begin with a lemma that reduces the problem to an earlier problem.

\begin{lemma}
Let $\cT$ be a phylogenetic $X$-tree, let $Y\subseteq X$, and let $\ell_{\max}$ be a non-negative real. If $\ell: E\rightarrow {\mathbb R}^{\ge 0}$ is a map on the set $E$ of edges of $\cT$ such that $\ell(e)\le \ell_{\max}$ for all $e\in E$, and $\varphi$ is a descendant diversity index, then
$$\Delta_{(\varphi, \cT, \ell_{\max})}(Y) \ge \Delta_{(\varphi,\cT, \ell)}(Y),$$
where $\ell_{\max}$ is the map $\ell_{\max}: E\rightarrow \mathbb R^{\ge 0}$ defined by $\ell_{\max}(e)=\ell_{\max}$ for all $e\in E$.
\label{lmax}
\end{lemma}

\begin{proof}
Let $e=(u, v)\in E$, and let $\lambda_{(\cT, \ell)}(e)$ and $\lambda_{(\cT, \ell_{\max})}(e)$ denote the contribution of $e$ to $\Delta_{(\varphi, \cT, \ell)}(Y)$ and $\Delta_{(\varphi, \cT, \ell_{\max})}(Y)$, respectively. To prove the lemma, it suffices to show that $\lambda_{(\cT, \ell)}(e)\le \lambda_{(\cT, \ell_{\max})}(e)$.

Since $\ell(e)\le \ell_{\max}$ and $\varphi$ is linear,
\begin{align*}
\lambda_{(\cT, \ell)}(e) & = \ell(e)-\sum_{y\in Y\cap C(v)} \gamma_{(\cT, \ell)}(y, e)\cdot \ell(e) \\
& = \ell(e)\cdot \left(1-\sum_{y\in Y\cap C(v)} \gamma_{(\cT, \ell)}(y, e)\right) \\
& \le \ell_{\max}\cdot \left(1-\sum_{y\in Y\cap C(v)} \gamma_{(\cT, \ell)}(y, e)\right) \\
& = \lambda_{(\cT, \ell_{\max})}(e),
\end{align*}
thereby completing the proof of the lemma.
\end{proof}

By Lemma~\ref{lmax}, we may assume that all edges of $\cT$ have length $\ell_{\max}$. Thus the problem is reduced to maximising the diversity difference on a fixed phylogenetic tree with fixed edge lengths. Therefore, we can find the maximum value in time $O\left(|X|^2\right)$ by applying Theorem~\ref{poly-time} and, in particular, {\sc MaxDiversityDiff}, to $(\cT, \ell_{\max})$, a non-negative integer $k$, and a descendant diversity index $\varphi$.

\section{Maximising the Diversity Difference Across All Trees}
\label{across-trees}

In contrast to the computational results in last section, in this section we characterise the edge-weighted phylogenetic $X$-trees that maximise the diversity difference for FP and ES. If the size of the subset of interest is $|X|$, the diversity difference is zero for all edge-weighted phylogenetic $X$-trees. Thus, throughout this section, we will impose the condition that $k\le |X|-1$.

\subsection{Maximum sum of edge lengths}

We first consider the problem of maximising the diversity difference across all edge-weighted phylogenetic trees whose edge lengths are constrained by their (total) sum. We begin with the Fair Proportion index.

\begin{theorem}
Let $k$ be a positive integer such that $k\le |X|-1$, and let $m$ be a positive real. Then an edge-weighted phylogenetic $X$-tree $(\cT, \ell)$ with $L(\cT, \ell)=m$ and a subset $Y$ of $X$ of size $k$ maximises
$$\Delta_{(FP, \cT, \ell)}(Y)=PD_{(\cT, \ell)}(Y)-\sum_{x\in Y} FP_{(\cT, \ell)}(x)$$
across all edge-weighted phylogenetic $X$-trees $(\cT', \ell')$ with $L(\cT', \ell')=m$ and all subsets of $X$ of size $k$ if and only if $(\cT, \ell)$ has an edge $e=(u, v)$ such that $\ell(e)=m$, $|Y\cap C(v)|=1$, and $X-Y\subseteq C(v)$, in which case,
$$\Delta_{(FP, \cT, \ell)}(Y)=m\left(1-\frac{1}{n-k+1}\right),$$
where $n=|X|$.
\label{across-sum-fp}
\end{theorem}

\begin{figure}
\center
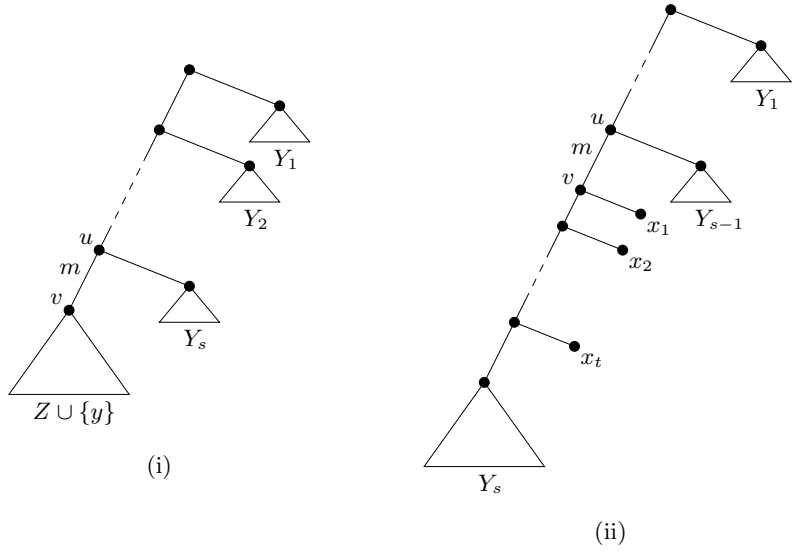
\caption{(i) An illustration of Theorem~\ref{across-sum-fp}, where {$X-Y = Z$, $Z\cup \{y\}=C(v)$, and $Y = Y_1\cup Y_2\cup \cdots \cup Y_s\cup \{y\}$}. All edges have weight zero except the edge $(u, v)$ which has weight $m$. The ``triangles'' represent pendant subtrees whose leaf sets are $Y_1, Y_2, \ldots, Y_s, Z\cup\{y\}$.
(ii) An illustration of Theorem~\ref{across-sum-es},  where $X-Y=\{x_1, x_2, \ldots, x_t\}$, $Y=Y_1\cup Y_2\cup \cdots \cup Y_s$, and {$Y_s$ is non-empty}. All edges have weight zero except the edge $(u, v)$ which has weight $m$. The ``triangles'' represent pendant subsets whose leaf sets are $Y_1, Y_2, \ldots, Y_s$.
}
\label{sum-fig}
\end{figure}

To illustrate Theorem~\ref{across-sum-fp}, a generic edge-weighted phylogenetic $X$-tree optimising $\Delta_{(FP, \cT, \ell)}(Y)$ as in the statement of this theorem is shown in Fig.~\ref{sum-fig}(i). 

\begin{proof}[Proof of Theorem~\ref{across-sum-fp}]
First note that if $(\cT, \ell)$ and $Y$ satisfy the conditions described in the statement of the theorem, then it follows that
$$\Delta_{(FP, \cT, \ell)}(Y)=m\left(1-\frac{1}{|X|-|Y|+1}\right)=m\left(1-\frac{1}{n-k+1}\right).$$
To show that the only if direction holds, let $(\cT, \ell)$ be an edge-weighted phylogenetic $X$-tree, and suppose that $(\cT, \ell)$ together with a subset $Y$ of size $k$ maximises $\Delta_{(FP, \cT, \ell)}(Y)$ across all edge-weighted phylogenetic $X$-trees $(\cT', \ell')$ with $L(\cT', \ell')=m$ and subsets of $X$ of size $k$. 
For each edge $e=(u, v)$ of $\cT$, 
let $Y_e=Y\cap C(v)$ and, for convenience, $X_e=C(v)=X\cap C(v)$. Thus
$$\lambda_{(\cT, \ell)}(e)=\ell(e)\left(1-\frac{|Y_e|}{|X_e|}\right)$$
if $|Y_e|\neq 0$ and $\lambda_{(\cT, \ell)}(e)=0$ if $|Y_e|=0$. Let $e_1, e_2, \ldots, e_t$ denote the edges of $\cT$ in which $\lambda_{(\cT, \ell)}(e_i) > 0$. Note that, for such edges, $X_e-Y_e$ is non-empty. We next show that $t=1$ by showing that if $t\geq 2$, then
$$\sum_{i\in \{1, 2, \ldots, t\}} \lambda_{(\cT, \ell)}(e_i) < m\left(1-\frac{1}{|X|-|Y|+1}\right),$$
contradicting the maximality of $\Delta_{(FP, \cT, \ell)}(Y)$.

Suppose $t\ge 2$. Since FP is descendant, it follows by Lemma~\ref{new-optimal}(i) that
$$\frac{|Y_{e_i}|}{|X_{e_i}|}\ = \frac{|Y_{e_j}|}{|X_{e_j}|}$$
for all $i, j\in \{1, 2, \ldots, t\}$. Therefore
\begin{align*}
\sum_{i\in \{1, 2, \ldots, t\}} \lambda_{(\cT, \ell)}(e_i) & = \ell(e_1)\left(1-\frac{|Y_{e_1}|}{|X_{e_1}|}\right) + \cdots + \ell(e_t)\left(1-\frac{|Y_{e_t}|}{|X_{e_t}|}\right) \\
& = \ell(e_1)\left(1-\frac{|Y_{e_1}|}{|X_{e_1}|}\right) + \cdots + \ell(e_t)\left(1-\frac{|Y_{e_1}|}{|X_{e_1}|}\right) \\
& \le m\left(1-\frac{|Y_{e_1}|}{|X_{e_1}|}\right)\\ 
& < m\left(1-\frac{|Y_{e_1}|}{|X|-|Y|+|Y_{e_1}|}\right) \\
& \le m\left(1-\frac{1}{|X|-|Y|+1}\right),
\end{align*}
where the second-to-last inequality holds {as the size of $X_{e_1}$ can be no more than the sum of the sizes of $X-Y$ and $Y\cap Y_{e_1}=Y_{e_1}$.} Moreover, it is strict when $t\ge 2$ since (i) if there is no directed path containing $e_1$ and $e_2$, then $X_{e_1}$ and $X_{e_2}$ are disjoint and, in particular, $X-(X_{e_1}\cup Y)$ is non-empty and (ii) if, without loss of generality, there is a directed path from $e_2$ to $e_1$, then $X_{e_1}\subset X_{e_2}$ and, as
$$\frac{|Y_{e_1}|}{|X_{e_1}|}=\frac{|Y_{e_2}|}{|X_{e_2}|},$$
again $X-(X_{e_1}\cup Y)$ is non-empty.

Hence $t=1$ and so, to maximise the diversity difference, $\ell(e_1)=m$, $|Y_{e_1}|=1$, and $X-Y\subseteq X_{e_1}$, in which case,
$$\Delta_{(FP, \cT, \ell)}(Y)=m\left(1-\frac{1}{n-k+1}\right).$$
This completes the proof of the theorem.
\end{proof}

The next theorem is the analogue of Theorem~\ref{across-sum-fp} for the Equal Splits index.

\begin{theorem}
Let $k$ be a positive integer such that $k\le |X|-1$, and let $m$ be positive real. Then an edge-weight phylogenetic $X$-tree $(\cT, \ell)$ with $L(\cT, \ell)=m$ and a subset $Y$ of $X$ of size $k$ maximises
$$\Delta_{(ES, \cT, \ell)}(Y)=PD_{(\cT, \ell)}(Y)-\sum_{x\in Y} ES_{(\cT, \ell)}(x)$$
across all phylogenetic $X$-trees $(\cT', \ell')$ with $L(\cT', \ell')=m$ and all subsets of $X$ of size $k$ if and only if $(\cT, \ell)$ has an edge $e=(u, v)$ such that $\ell(e)=m$, $X-Y$ is a chain whose first parent is $v$, and $Y\cap C(v)$ is non-empty, in which case,
$$\Delta_{(ES, \cT, \ell)}(Y)=m\left(1-\frac{1}{2^{n-k}}\right),$$
where $n=|X|$.
\label{across-sum-es}
\end{theorem}

A generic edge-weighted phylogenetic $X$-tree optimising $\Delta_{(ES, \cT, \ell)}(Y)$ as in the statement of Theorem~\ref{across-sum-es} is shown in Fig.~\ref{sum-fig}(ii).
The proof of Theorem~\ref{across-sum-es} takes a similar approach to that of Theorem~\ref{across-sum-fp} but is slightly more involved. We begin with a lemma.

\begin{lemma}
Let $\cT$ be a phylogenetic $X$-tree, and let $Y\subseteq X$. Let $e=(u, v)$ be an edge of $\cT$, and suppose that $(X-Y)\cap C(v)$ and $Y\cap C(v)$ are both non-empty. Let $\cT'$ be the phylogenetic $X$-tree obtained from $\cT$ by replacing the pendant subtree $\cT|C(v)$ with a phylogenetic tree on leaf set $C(v)$, where $C(v)-Y$ is a chain whose first parent is $v$ and $C(v)\cap Y$ is a pendant subtree below the chain. If $\cT$ is not of this form, then
$$\frac{1}{2^{n_1-k_1}}=\sum_{y\in Y\cap C(v)} \frac{1}{2^{|P(\cT'; v, y)|}} < \sum_{y\in Y\cap C(v)} \frac{1}{2^{|P(\cT; v, y)|}},$$
where $n_1=|C(v)|$ and $k_1=|Y\cap C(v)|$.
\label{chain}
\end{lemma}

\begin{proof}
Suppose that $\cT$ is not of the same form as $\cT'$. Then there is a $z\in C(v)-Y$ such that either $z$ is in a cherry $\{z, z'\}$ of $\cT$, where $z'\not\in Y$ or $z$ is a leaf of a pendant subtree of $\cT$ whose two maximal subtrees each contain an element of $Y$. Let $\cT''$ be the phylogenetic $X$-tree that is obtained from $\cT$ by pruning $z$ and regrafting to the edge $e$. Relabel the vertex $v$ as $v'$ and the newly created vertex in the subdivided edge $v$ (so that $C(v)$ is unchanged). It is now easily checked that if $y\in Y\cap C(v)$, then $|P(\cT; v, y)|\le |P(\cT''; v, y)|$. Moreover, for some $y'\in Y\cap C(v)$, we have $|P(\cT; v, y')| < |P(\cT''; v, y')|$. In particular,
$$\sum_{y\in Y\cap C(v)} \frac{1}{2^{|P(\cT''; v, y)|}} < \sum_{y\in Y\cap C(v)} \frac{1}{2^{|P(\cT; v, y)|}}.$$
The lemma now follows by repeating this process until we have constructed a phylogenetic tree in the same form as $\cT'$.
\end{proof}

\begin{proof}[Proof of Theorem~\ref{across-sum-es}]
{If $(\cT, \ell)$ and $Y$ satisfy the conditions described in the statement of Theorem~\ref{across-sum-es}, then it is easily checked that
$$\Delta_{(ES, \cT, \ell)}(Y)=m\left(1-\frac{1}{2^{n-k}}\right).$$
To complete the proof}, let $(\cT, \ell)$ be an edge-weighted phylogenetic $X$-tree, and suppose that $(\cT, \ell)$ together with a subset $Y$ of $X$ of size $k$ maximises $\Delta_{(ES, \cT, \ell)}(Y)$ across all edge-weighted phylogenetic $X$-trees $(\cT', \ell')$ with $L(\cT', \ell')=m$ and subsets of $X$ of size $k$. For each edge $e=(u, v)$ of $\cT$, define $\lambda_{(\cT, \ell)}(e)$ to be the contribution of $e$ to $\Delta_{(ES, \cT, \ell)}(Y)$. Thus
$$\lambda_{(\cT, \ell)}(e)=\ell(e)\left(1-\sum_{y\in Y\cap C(v)} \frac{1}{2^{|P(\cT; v, y)|}}\right)$$
if $|Y\cap C(v)|\neq 0$ and $\lambda_{(\cT, \ell)}(e)=0$ if $|Y\cap C(v)|=0$. Note that, if $\lambda_{(\cT, \ell)}(e) > 0$, then $C(v)-Y$ is non-empty. Let $e_1=(u_1, v_1), e_2=(u_2, v_2), \ldots, e_t=(u_t, v_t)$ denote the edges of $\cT$ in which $\lambda_{(\cT, \ell)}(e_i)>0$.
We next show that $t=1$ by showing that if $t>1$ then 
$$m\left(1-\frac{1}{2^{n-k}}\right) > \sum_{i\in \{1, 2, \ldots, t\}} \lambda_{(\cT, \ell)}(e_i).$$

Say $t\ge 2$. Since ES is descendant, it follows by Lemma~\ref{new-optimal}(i) that
\begin{align}
\sum_{y\in Y\cap C(v_i)} \frac{1}{2^{|P(\cT; v_i, y)|}} = \sum_{y\in Y\cap C(v_j)} \frac{1}{2^{|P(\cT; v_j, y)|}}
\label{equality}
\end{align}
for all $i, j\in \{1, 2, \ldots, t\}$. Hence
\begin{align*}
\sum_{i\in \{1, 2, \ldots, t\}} \lambda_{(\cT, \ell)}(e_i)
& = \ell(e_1)\left(1-\sum_{y\in Y\cap C(v_1)} \frac{1}{2^{|P(\cT; v_1, y)|}}\right) + \cdots \\
& \qquad \qquad \qquad \qquad \cdots + \ell(e_t)\left(1-\sum_{y\in Y\cap C(v_t)} \frac{1}{2^{|P(\cT; v_t, y)|}}\right) \\
& = \ell(e_1)\left(1-\sum_{y\in Y\cap C(v_1)} \frac{1}{2^{|P(\cT; v_1, y)|}}\right) + \cdots \\
& \qquad \qquad \qquad \qquad \cdots + \ell(e_t)\left(1-\sum_{y\in Y\cap C(v_1)} \frac{1}{2^{|P(\cT; v_1, y)|}}\right) \\
& \le m\left(1-\sum_{y\in Y\cap C(v_1)} \frac{1}{2^{|P(\cT; v_1, y)|}}\right) \\
& \le m\left(1-\frac{1}{2^{n_1-k_1}}\right),
\end{align*}
where $n_1=|C(v_1)|$ and $k_1=|Y\cap C(v_1)|$, and the last inequality holds by Lemma~\ref{chain}. Since $t\ge 2$, it follows by (\ref{equality}) that $X-(C(v_1)\cup Y)$ is non-empty, so $n-n_1>0$. That is,
$$\sum_{i\in \{1, 2, \ldots, t\}} \lambda_{(\cT, \ell)}(e_i)\le m\left(1-\frac{1}{2^{n_1-k_1}}\right) < m\left(1-\frac{1}{2^{n-k_1}}\right).$$
Hence $t=1$, and so to maximise the diversity difference, a single edge $e=(u, v)$ say of $(\cT, \ell)$ has length $\ell(e)=m$ and, by Lemma~\ref{chain}, $X-Y$ is a chain with first parent $v$, and $Y\cap C(v)$ is non-empty, in which case,
$$\Delta_{(ES, \cT, \ell)}(Y)=m\left(1-\frac{1}{2^{n-k}}\right).$$
This completes the proof of the theorem.
\end{proof}

\subsection{Maximum edge length}

We next consider the problem of maximising the diversity difference across all phylogenetic trees whose edge lengths are constrained by some maximum value. By Lemma~\ref{lmax}, we may assume all edges of $\cT$ have the same weight. Thus, without loss of generality, we will assume that all edges have weight~$1$ and so, for simplicity, we write $\cT$ for $(\cT, \ell)$. {To this end, for a descendant diversity index $\varphi$, the contribution of an edge $e$ of $\cT$ to $\Delta_{(\varphi, \cT)}(Y)$ is denoted by $\lambda_{\cT}(e)$.} We begin with two lemmas and a corollary.

\begin{lemma}
Let $\cT$ be a phylogenetic $X$-tree, let $Y\subseteq X$ such that $|Y|\le |X|-1$, and let $\varphi\in \{FP, ES\}$. If $C$ is a non-trivial cluster of $\cT$ such that either $C\subseteq Y$ or $C\subseteq X-Y$, then there exists a phylogenetic $X$-tree $\cT'$ such that
$$\Delta_{(\varphi, \cT')}(Y) > \Delta_{(\varphi, \cT)}(Y).$$
\label{bichromatic1}
\end{lemma}

\begin{proof}
Let $v$ be a non-leaf vertex of $\cT$, and suppose that $C(v)\subseteq Y$ or $C(v)\subseteq (X-Y)$. Without loss of generality, we may assume that $C(v)$ is maximal with this property. Let $(u, v)$ be the edge of $\cT$ directed into $v$ and note that, for $(u, v)$ and all edges on a path from $v$ to a leaf, the contribution of these edges under either FP or ES to $\Delta_{(\varphi, \cT)}(Y)$ is zero regardless of whether $C(v)\subseteq Y$ or $C(v)\subseteq (X-Y)$. Let $w$ be the child vertex of $u$ that is not $v$. By maximality and $|Y|\neq |X|$, if $C(v)\subseteq Y$, then $(X-Y)\cap C(w)\neq \emptyset$ and, if $C(v)\subseteq (X-Y)$, then $Y\cap C(w)\neq \emptyset$. Let $E$ denote the edge set of $\cT$, let $P$ denote the set of edges of $\cT$ on the path from its root $\rho$ to $u$, and let $a\in C(v)$. Depending on whether $C(v)\subseteq Y$ or $C(v)\subseteq (X-Y)$, we will compare $\cT$ with two other phylogenetic $X$-trees.

Let $\cT_1$ denote the phylogenetic $X$-tree obtained from $\cT$ by pruning $a$ and regrafting it to $\rho$. Let $E_1$ denote the set of edges of $\cT_1$, let $\rho_1$ denote the root of $\cT_1$, let $q$ denote the child of $\rho_1$ that is not $a$, and let $P_1$ denote the set of edges of $\cT_1$ on the path from $q$ to $u$. Let $\cT_2$ denote the phylogenetic $X$-tree obtained from $\cT$ by pruning $a$ and regrafting it to $(u, w)$. Let $E_2$ denote the set of edges of $\cT_2$, let $p_a$ denote the parent of $a$ in $\cT_2$, and let $P_2$ denote the set of edges of $\cT_2$ on the path from its root to $p_a$.

First suppose that $\varphi$ is FP. If $C(v)\subseteq Y$, then
$$\sum_{e\in E-P} \lambda_{\cT}(e)=\sum_{e_1\in E_1-(P_1\cup \{(\rho_1, q)\})} \lambda_{\cT_1}(e_1).$$
Furthermore, as $a\in Y$, and $C_{\cT_1}(u)\cap Y$ and $C_{\cT_1}(u)\cap (X-Y)$ are both non-empty,
$$\sum_{e\in P} \lambda_{\cT}(e) < \sum_{e_1\in P_1} \lambda_{\cT_1}(e_1)$$
and $\lambda_{\cT_1}((\rho_1, q)) > 0$. Hence $\Delta_{(FP, \cT_1)}(Y) > \Delta_{(FP, \cT)}(Y)$. On the other hand, if $C(v)\subseteq (X-Y)$, then it is easily checked that
$$\sum_{e\in E} \lambda_{\cT}(e) = \sum_{e_2\in E_2-\{(u, p_a)\}} \lambda_{\cT_2}(e_2)$$
and
$$\lambda_{\cT_2}((u, p_a)) > 0$$
as $C_{\cT_2}(p_a)$ has a non-empty intersection with $Y$ and $X-Y$. Thus $\Delta_{(FP, \cT_2)}(Y) > \Delta_{(FP, \cT)}(Y)$, and so the lemma holds if $\varphi$ is FP.

Now suppose that $\varphi$ is ES. If $C(v)\subseteq Y$, then
$$\sum_{e\in E} \lambda_{\cT}(e) = \sum_{e_1\in E_1-\{(\rho_1, q)\}} \lambda_{\cT_1}(e_1).$$
Since $C_{\cT_1}(q)\cap Y$ and $C_{\cT_1}(q)\cap (X-Y)$ are both non-empty, $\lambda_{\cT_1}((\rho_1, q)) > 0$, and so $\Delta_{(ES, \cT_1)}(Y) > \Delta_{(ES, \cT)}(Y)$ if $\varphi$ is ES. Furthermore, if $C(v)\subseteq (X-Y)$, then
$$\sum_{e\in E-P} \lambda_{\cT}(e) = \sum_{e_2\in E_2-P_2} \lambda_{\cT_2}(e_2).$$
Hence, as $a\in (X-Y)$, 
\begin{align*}
\Delta_{(ES, \cT_2)}(Y)- & \Delta_{(ES, \cT)}(Y) \\
& = \sum_{e_2\in P_2 } \lambda_{\cT_2}(e_2)-\sum_{e\in P} \lambda_{\cT}(e) \\
& = \sum_{e_2\in P_2-\{(u, p_a)\}}\left({1-}\sum_{y\in Y\cap C(w)} \frac{1}{2^{|P(\cT_2; e_2, y)|}}\right) + \lambda_{\cT_2}((u, p_a)) \\
& \qquad \qquad \qquad {-} \sum_{e\in P}\left({1-}\sum_{y\in Y\cap C(w)} \frac{1}{2^{|P(\cT; e, y)|}}\right) \\
& = -\tfrac{1}{2}\cdot \left(\sum_{e\in P}\left(\sum_{y\in Y\cap C(w)} \frac{1}{2^{|P(\cT; e, y)|}}\right)\right) + \lambda_{\cT_2}((u, p_a)) \\
& \qquad \qquad \qquad + \sum_{e\in P}\left(\sum_{y\in Y\cap C(w)} \frac{1}{2^{|P(\cT; e, y)|}}\right) \\
& = \tfrac{1}{2}\cdot \left(\sum_{y\in Y\cap C(w)} \frac{1}{2^{|P(\cT; e, y)|}}\right) + \lambda_{\cT_2}((u, p_a)) \\
& > 0
\end{align*}
as $\lambda_{\cT_2}((u, p_a)) > 0$ since $C_{\cT_2}(p_a)\cap Y$ and $C_{\cT_2}(p_a)\cap (X-Y)$ are both non-empty. Therefore $\Delta_{(ES, \cT_2)}(Y) > \Delta_{(ES, \cT)}(Y)$. This completes the proof of the lemma.
\end{proof}

An immediate consequence of Lemma~\ref{bichromatic1} is the next corollary.

\begin{corollary}
Let $\cT$ be a phylogenetic $X$-tree, let $k$ be a positive integer such that $k\le |X|-1$, let $Y$ be a subset of $X$ of size $k$, and let $\varphi\in \{FP, ES\}$. Suppose that $\cT$ and $Y$ maximises
$$\Delta_{(\varphi, \cT)}(Y)=PD_{\cT}(Y)-\sum_{x\in Y} \varphi_{\cT}(x)$$
across all phylogenetic $X$-trees and all subsets of $X$ of size $k$. If $\{a, b\}$ is a cherry of $\cT$, then $|\{a, b\}\cap Y|=1$. In particular, $\cT$ has no non-trivial cluster $C$ such that either $C\subseteq Y$ or $C\subseteq (X-Y)$.
\label{bichromatic2}
\end{corollary}

\begin{lemma}
Let $\cT$ be a phylogenetic $X$-tree, let $k$ be a positive integer such that $k\le |X|-1$, let $Y$ be a subset of $X$ of size $k$, and let $\varphi\in \{FP, ES\}$. Suppose that $\cT$ and $Y$ maximises
$$\Delta_{(\varphi, \cT)}(Y)=PD_{\cT}(Y)-\sum_{x\in Y} \varphi_{\cT}(x)$$
across all phylogenetic $X$-trees and all subsets of $X$ of size $k$. If $(a_1, a_2, \ldots, a_s)$ is a chain of $\cT$, in which $\{a_1, a_2\}$ is a cherry, then
\begin{enumerate}[{\rm (i)}]
\item $|\{a_1, a_2\}\cap Y|=1$ and 

\item for some $p\ge 3$, we have $|\{a_1, a_2, \ldots, a_{p-1}\}\cap Y|=1$ and $\{a_p, a_{p+1}, \ldots, a_s\}\subseteq Y$.
\end{enumerate}
\label{1-chain}
\end{lemma}

\begin{proof}
By Corollary~\ref{bichromatic2}, $(a_1, a_2, \ldots, a_s)$ satisfies (i). Suppose that $(a_1, a_2, \ldots, a_s)$ does not satisfy (ii). Let $i\ge 3$ be the least index such $a_i\in Y$ but $a_{i+1}\in X-Y$. Let $\cT'$ be the phylogenetic $X$-tree obtained from $\cT$ by interchanging $a_i$ and $a_{i+1}$, that is, pruning $a_i$ and regrafting it to the edge directed into the parent of $a_{i+1}$ so that $(a_1, a_2, \ldots, a_{i-1}, a_{i+1}, a_i, a_{i+2}, \ldots, a_s)$ is a chain of $\cT'$. Let $p_i$ and $p_{i+1}$ denote the parents of $a_i$ and $a_{i+1}$ in $\cT$, respectively, and let $p'_i$ and $p'_{i+1}$ denote the parents of $a_i$ and $a_{i+1}$ in $\cT'$, respectively.

If $\varphi$ is the Fair Proportion index, then
\begin{align*}
\Delta_{(FP, \cT')}(Y)-\Delta_{(FP, \cT)}(Y) & = \lambda_{\cT'}((p'_i, p'_{i+1}))-\lambda_{\cT}((p_{i+1}, p_i)) \\
& = {\textstyle \left(1-\frac{1}{i}\right) - \left(1-\frac{2}{i}\right)} \\
& > 0.
\end{align*}
This contradiction to the maximality of $\cT$ and $Y$ implies that the lemma holds when $\varphi$ is FP.

Now suppose that $\varphi$ is the Equal Splits index. Let $P$ denote the path in $\cT$ from the root to $p_i$ and let $P'$ denote the path in $\cT'$ from the root to $p'_i$. Then
\begin{align*}
\Delta_{(ES, \cT')}(Y)- \Delta_{(ES, \cT)}(Y) & = -\varphi_{\cT'}(a_i) + \varphi_{\cT}(a_i) \\
& = -\sum_{e'\in P'} \frac{1}{2^{|P(\cT'; e', a_i)|}} + \sum_{e\in P} \frac{1}{2^{|P(\cT; e, a_i)|}}.
\end{align*}
For all $e\in P-\{(p_{i+1}, p_i)\}$, we have $|P(\cT'; e, a_i)|=|P(\cT; e, a_i)|-1$, and the contribution of $(p_{i+1}, p_i)$ to $ES_{\cT}(a_i) $ is $\tfrac{1}{2}$. So
\begin{align*}
\Delta_{(ES, \cT')}(Y) - \Delta_{(ES, \cT)}(Y) & = -2\sum_{e\in P-\{(p_{i+1}, p_i)\}} \frac{1}{2^{|P(\cT; e, a_i)|}} \\
& \qquad \qquad \quad + \left(\tfrac{1}{2} + \sum_{e\in P-\{(p_{i+1}, p_i)\}} \frac{1}{2^{|P(\cT; e, a_i)|}}\right) \\
& = \tfrac{1}{2} -\sum_{e\in P-\{(p_{i+1}, p_i)\}} \frac{1}{2^{|P(\cT; e, a_i)|}} \\
& \ge \tfrac{1}{2} - {\textstyle \left(\frac{1}{4}+\frac{1}{8}+\cdots +\frac{1}{2^{|P|}}\right)} \\
& > 0.
\end{align*}
This contradiction to the maximality of $\cT$ and $Y$ implies that the lemma holds when $\varphi$ is ES, thereby completing the proof of the lemma.
\end{proof}

The next theorem is illustrated in Fig.~\ref{max-fig}. In particular, an edge-weighted phylogenetic $X$-tree optimising $\Delta_{(\varphi, \cT, \ell)}(Y)$ as in the statement of this theorem is shown in Fig.~\ref{max-fig}, where $\varphi\in \{FP, ES\}$. It is interesting to note that the outcome of Theorem~\ref{across-max-fp} is the same for FP and ES, although the proof of the theorem requires FP and ES to be considered separately.

\begin{figure}
\center
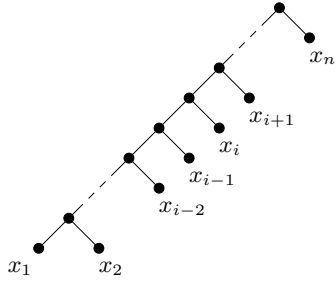
\caption{An illustration of Theorem~\ref{across-max-fp}, where $Y=\{x_1,x_i,x_{i+1},\ldots,x_n\}$, so $|\{x_1, x_2\}\cap Y|=1$, {and $X-Y=\{x_2, x_3, \ldots, x_{i-1}\}$}.}
\label{max-fig}
\end{figure}

\begin{theorem}
Let $k$ be a positive integer such that $k\le |X|-1$, let $\ell_{\max}$ be a positive real, and let $\varphi\in \{FP, ES\}$. Then an edge-weighted phylogenetic $X$-tree $(\cT, \ell)$ with $\ell(e)\le \ell_{\max}$ for all $e\in E(\cT)$ and a subset $Y$ of $X$ of size $k$ maximises
$$\Delta_{(\varphi, \cT, \ell)}(Y)=PD_{(\cT, \ell)}(Y)-\sum_{x\in Y} \varphi_{(\cT, \ell)}(x)$$
across all phylogenetic $X$-trees $(\cT', \ell')$ with $\ell'(e)\le \ell_{\max}$ for all $e\in E(\cT')$ and all subsets of $X$ of size $k$ if and only if $\ell(e)=\ell_{\max}$ for all $e\in E(\cT)$ and $X$ is a chain $(x_1, x_2, \ldots, x_n)$ such that $|\{x_1, x_2\}\cap Y|=1$, and, for some $i\ge 3$, $|\{x_1, x_2, \ldots, x_{i-1}\}\cap Y|=1$ and $\{x_i, x_{i+1}, \ldots, x_n\}\subseteq Y$.
\label{across-max-fp}
\end{theorem}

\begin{proof}
Let $(\cT, \ell)$ be an edge-weighted phylogenetic $X$-tree with $\ell(e)\le \ell_{\max}$ for all $e\in E(\cT)$ and let $Y$ be a subset of $X$ of size $k$. Suppose that $(\cT, \ell)$ and $Y$ maximise $\Delta_{(\varphi, \cT, \ell)}(Y)$ across all edge-weighted phylogenetic $X$-trees $(\cT', \ell')$ with $\ell'(e)\le \ell_{\max}$ for all $e\in E(\cT')$ and all subsets of $X$ of size $k$. {Recall that, by Lemma~\ref{lmax}, we may assume that all edges have weight $\ell_{\max}=1$.}

We first show that $\cT$ does not have two distinct cherries. Suppose that $\cT$ has two such cherries. Then $\cT$ has a vertex $v$ in which $C(v)$ contains exactly two cherries. By Lemma~\ref{1-chain} and the maximality of $(\cT, \ell)$ and $Y$, the set $C(v)$ consists of two disjoint chains $(a_1, a_2, \ldots, a_s)$ and $(b_1, b_2, \ldots, b_t)$ each of which satisfies properties (i) and (ii) in the statement of Lemma~\ref{1-chain}. Let $A=\{a_1, a_2, \ldots, a_s\}$ and let $B=\{b_1, b_2, \ldots, b_t\}$.

\begin{figure}
\center
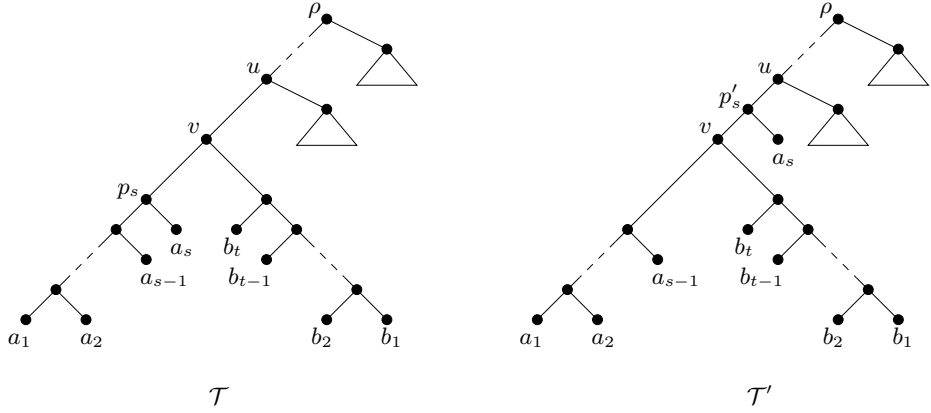
\caption{Illustrating the pruning and regrafting of the leaf $a_s$ in the proof of Theorem~\ref{across-max-fp} for FP.}
\end{figure}

First suppose that $\varphi$ is FP. Without loss of generality, we may assume that
$$\frac{|Y\cap A|}{s}\ge \frac{|Y\cap B|}{t}.$$
Say $s\ge 3$ and $a_s\in Y$. Let $\cT'$ be the phylogenetic $X$-tree obtained from $\cT$ by pruning $a_s$ and regrafting it to the edge $(u, v)$ directed into $v$. {This operation is shown in Fig.~\ref{max-fig}(ii), where $\cT$ is drawn above $\cT'$.} Let $p_s$ and $p'_s$ denote the parents of $a_s$ in $\cT$ and $\cT'$, respectively. Then it is easily seen that
$$\sum_{e\in E(\cT)-\{(v, p_s)\}} \lambda_{\cT}(e)=\sum_{e'\in E(\cT')-\{(p'_s, v)\}} \lambda_{\cT'}(e').$$
Furthermore, as $|Y\cap A|/s\ge |Y\cap B|/t$,
\begin{align*}
\lambda_{\cT}((v, p_s)) & = 1-\frac{|Y\cap A|}{s} \\
& \le 1-\frac{|Y\cap (A\cup B)|}{s+t} \\
& < 1-\frac{|Y\cap (A\cup B)|-1}{s+t-1} \\
& = \lambda_{\cT'}((p'_s, v)).
\end{align*}
Thus $\Delta_{(FP, \cT)}(Y) < \Delta_{(FP, \cT')}(Y)$. It follows that $|Y\cap A|=1$, and so we may assume that $Y\cap A=\{a_1\}$.

Say $t\ge 3$ and $b_t\in Y$. Now let $\cT'$ be the phylogenetic $X$-tree obtained from $\cT$ by pruning $b_t$ and regrafting it to $(u, v)$. let $p_t$ and $p'_t$ denote the parents of $b_t$ in $\cT$ and $\cT'$, respectively. Then
$$\sum_{e\in E(\cT)-\{(v, p_t)\}} \lambda_{\cT}(e) = \sum_{e'\in E(\cT')-\{(p'_t, v)\}} \lambda_{\cT'}(e').$$
Also, as $|Y\cap A|=1$ and $|Y\cap B|\ge 2$,
\begin{align*}
\lambda_{\cT}((v, p_t)) & = 1-\frac{|Y\cap B|}{t} \\
& < 1-\frac{|Y\cap B|}{s+t} \\
& \le 1-\frac{|Y\cap (A\cup B)|-1}{s+t-1} \\
& = \lambda_{\cT'}((p'_t, v)).
\end{align*}
Thus $\Delta_{(FP,\cT)}(Y) < \Delta_{(FP,\cT')}(Y)$, contradicting the maximality of $(\cT, \ell)$ and $Y$. Therefore $|Y\cap B|=1$, and so we may assume that $Y\cap (A\cup B)=\{a_1, b_1\}$.

Without loss of generality, we may now assume that $t\ge s$. Let $p_s$ and $q_t$ denote the parents of $a_s$ and $b_t$ in $\cT$, respectively. Let $\cT''$ be the phylogenetic $X$-tree obtained from $\cT$ by pruning $a_s$ and regrafting it to the edge $(u, q_t)$. Let $p''_s$ denote the parent of $a_s$ in $\cT''$. Then
$$\sum_{e\in E(\cT)-\{(v, p_s)\}} \lambda_{\cT}(e) = \sum_{e''\in E(\cT'')-\{(v, p''_s)\}} \lambda_{\cT''}(e'')$$
and
\begin{align*}
\lambda_{\cT}((u, p_s)) & = {\textstyle 1-\frac{1}{s}} \\
& < {\textstyle 1-\frac{1}{t+1}} \\
& = \lambda_{\cT''}((u, p''_s))
\end{align*}
as $t\ge s$. Repeating this process for each of $a_{s-1}, a_{s-2}, \ldots, a_2$, we obtain a phylogenetic tree $X$-tree $\cT_1$ with one less cherry than $\cT$, and
$$\Delta_{(FP, \cT)}(Y) < \Delta_{(FP, \cT_1)}(Y),$$
a contradiction to maximality. We deduce that $\cT$ has exactly one cherry if $\varphi$ is FP.

\begin{figure}
\center
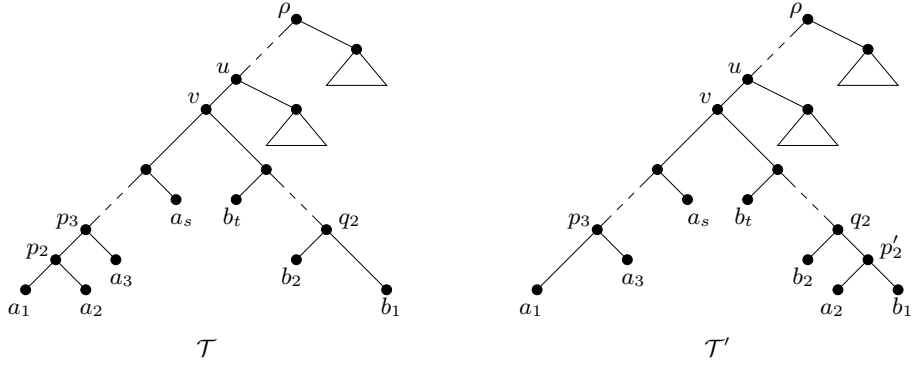
\caption{{Illustrating the pruning and regrafting of the leaf $a_2$ in the proof of Theorem~\ref{across-max-fp} for ES.}}
\label{max-fig2}
\end{figure}

Now suppose that $\varphi$ is ES. Without loss of generality, we may assume that $t\ge s$ and $a_1, b_1\in Y$. Let $\cT'$ be the phylogenetic $X$-tree obtained from $\cT$ by pruning $a_2$ and regrafting it to the edge $(q_2, b_1)$, where $q_2$ is the parent of $b_2$ in $\cT$ (as depicted in Fig.~\ref{max-fig2}). Then, noting that $a_2$ is not in $Y$,
$$\sum_{y\in (Y-\{a_1, b_1\})} ES_{\cT}(y) = \sum_{y\in (Y-\{a_1, b_1\})} ES_{\cT'}(y).$$

Furthermore, since the length of the each path from an edge to $a_1$ is reduced by $1$ in forming $\cT'$ {as} edge $(p_3, p_2)$ is {suppressed},
{$$\sum_{e\in P(\cT; \rho, a_1)} \frac{1}{2^{|P(\cT; e, a_1)|}} - \sum_{e'\in P(\cT'; \rho, a_1)} \frac{1}{2^{|P(\cT'; e', a_1)|}} = \frac{1}{2^{|P(\cT; \rho, a_1)|-1}}$$}
and. since the length of the each path from an edge to $b_1$ is increased by 1 in forming $\cT'$ {as} edge $(q_2, {p'_2})$ is added,
{$$\sum_{e'\in P(\cT'; \rho, b_1)} \frac{1}{2^{|P(\cT'; e', b_1)|}} - \sum_{e\in P(\cT; \rho, b_1)} \frac{1}{2^{|P(\cT; e, b_1)|}} = \frac{1}{2^{|P(\cT'; \rho, b_1)|-1}}.$$}
{But, as $t\ge s$, 
$$\frac{1}{2^{|P(\cT'; \rho, b_1)|-1}} < \frac{1}{2^{|P(\cT; \rho_1, a_1)|-1}}.$$
Therefore
\begin{align*}
\sum_{y\in \{a_1, b_1\}} ES_{\cT}(y) - \sum_{y\in \{a_1, b_1\}} ES_{\cT'}(y) & = \frac{1}{2^{|P(\cT; \rho, a_1)|-1}} - \frac{1}{2^{|P(\cT'; \rho, b_1)|-1}} \\
& > 0.
\end{align*}
It now follows that $\Delta_{(ES, \cT)}(Y) < \Delta_{(ES, \cT')}(Y)$.} This contradiction to maximality implies that $\cT$ has exactly one cherry if $\varphi$ is ES. It now follows by Lemma~\ref{1-chain} that, if $\varphi\in \{FP, ES\}$, then to maximise $\Delta_{(\varphi, \cT)}(Y)$ we have that $X$ is a chain and $Y$ is a subset of $X$ of size $k$ as described in the statement of the theorem.
\end{proof}

\section{Discussion}
\label{discussion}

The results for Fair Proportion and Equal Splits in the last section are strikingly similar. Indeed, they are exactly the same when considering the outcomes of Theorem~\ref{across-max-fp}. While the proof of Theorem~\ref{across-max-fp} eventually required separating into two parts to independently consider FP and ES, the two approaches taken were alike. This suggests that there is probably a natural class of phylogenetic diversity indices which these results are representative of. {If so, what is this class? The class of descendant diversity indices is unlikely to be sufficient. However, what if we additionally required the following property? Let $(\cT, \ell)$ be an edge-weighted phylogenetic $X$-tree and let $\varphi$ be a descendant diversity index. Let $e=(u, v)$ be an edge of $\cT$, and let $\{x, y\}\subseteq X$ such that $\{x, y\}\subseteq C(v)$. If the number of edges from $v$ to $y$ is at most the number of edges from $v$ to $x$, then the contribution of $\ell(e)$ to $\varphi_{(\cT, \ell)}(x)$ is at most the contribution of $\ell(e)$ to $\varphi_{(\cT, \ell)}(y)$.}

\end{document}